\documentclass[final,5p,twocolumn]{elsarticle}

\usepackage{graphicx}
\usepackage{amsmath}
\usepackage{amsthm}
\usepackage{amssymb}
\usepackage{mathtools}	
\usepackage{times}
\usepackage{enumitem}
\usepackage{hyperref}
\usepackage{booktabs}
\usepackage{appendix}
\usepackage[usenames,dvipsnames,svgnames,table]{xcolor}
\usepackage{amscd,graphicx,array,dsfont,texdraw,tikz,enumerate,multirow,verbatim}

\usepackage{caption}
\usepackage{subcaption}
\usepackage{cleveref}
\setcitestyle{authoryear,open={(},close={)}} 

\hypersetup{
        colorlinks = true,
	linkcolor={red!80!black},
	citecolor={green!80!black},
	urlcolor={blue!80!black}
}

\newtheorem{lem}{Lemma}[section]
\newtheorem{prop}{Proposition}[section]

\newtheorem{prob}{Problem}[section]
\newtheorem{assum}{Assumption}[section]

\newtheorem{rem}{Remark}[section]

\renewcommand{\appendix}{\par
  \setcounter{section}{0}
  \setcounter{subsection}{0}
  \gdef\thesection{\Alph{section}}
}

\newcommand{\bfA}{\mathbf{A}}
\newcommand{\bfC}{\mathbf{C}}
\newcommand{\bfT}{\mathbf{T}}

\newcommand{\R}{\mathbb{R}}
\renewcommand{\t}{\mathsf{T}}

\newcommand{\bbm}[1]{\left[\begin{matrix} #1 \end{matrix}\right]}
\newcommand{\sbm}[1]{\left[\begin{smallmatrix} #1
   \end{smallmatrix}\right]}

\usetikzlibrary{arrows,shapes,shadows,positioning,automata,patterns}
\usetikzlibrary{trees,decorations.pathmorphing,decorations.markings}
\usetikzlibrary{shapes.geometric,backgrounds,calc}
\usepackage{tikz}
\usetikzlibrary{er,positioning}
\usetikzlibrary{matrix,calc,shapes} 
\tikzstyle{block} = [draw,rectangle,rounded corners=0.5mm,thick,minimum height=0.8cm,minimum width=0.8cm,fill=gray!10,draw=black!50,align=center]
\tikzstyle{blocke} = [draw,rectangle,rounded corners=0.5mm,thick,minimum height=0.7cm,minimum width=0.8cm,fill=gray!10,draw=black!50,align=center]
\tikzstyle{sum} = [draw,circle,inner sep=0mm,minimum size=2mm,thick,fill=white!40,draw=black!50]
\tikzstyle{dot} = [draw,circle,inner sep=0mm,minimum size=0.5pt,thick,fill=black,draw=black]

\journal{Automatica}

\begin{document}

\begin{frontmatter}

\title{{\Large\bf Velocity-free task-space regulator for robot manipulators with external disturbances
}}

\tnotetext[t1]{The work was supported in part by the Natural Science Foundation of China under Grant No. 62073168, in part by the NXTGEN High-Tech program on Autonomous Factory, and in part by the InnoHK Clusters of the Hong Kong SAR Government via the Hong Kong Centre for Logistics Robotics}

\author[cuhk]{Haiwen Wu}\ead{haiwenwu@cuhk.edu.hk}
\author[rug]{Bayu Jayawardhana}\ead{b.jayawardhana@rug.nl}
\author[scut]{Dabo Xu}\ead{dxu@scut.edu.cn}
\address[cuhk]{Department of Mechanical and Automation Engineering, The Chinese University of Hong Kong, Hong Kong Special Administrative Region}
\address[rug]{Engineering and Technology Institute Groningen, Faculty of Science and Engineering, University of Groningen, Groningen 9747 AG, the Netherlands}
\address[scut]{Shien-Ming Wu School of Intelligent Engineering, South China University of Technology, Guangzhou 511442, China.}  	


\begin{abstract}   
This paper addresses the problem of task-space robust regulation of robot manipulators subject to external disturbances. A velocity-free control law is proposed by combining the internal model principle and the passivity-based output-feedback control approach. The resulting controller not only ensures asymptotic convergence of the regulation error but also rejects unwanted external sinusoidal disturbances. The potential of the proposed method lies in its simplicity, intuitiveness, and straightforward gain selection criteria for the synthesis of multi-joint robot manipulator control systems. 
\end{abstract}
                          
\end{frontmatter}

\section{Introduction}\label{sec1}
Control of multi-joint robotic systems has been an active research area in both
robotics and control communities for over three decades. Among various
challenges, an interesting topic involves effectively mitigating external
disturbances and/or measurement noise to achieve high-precision control
performance. In many mechanical control applications, systems often encounter sinusoidal or periodic disturbances arising from rotational elements such as motors and vibratory components~\cite{Zarikian2007harmonic,Tomizuka2008dealing}. The presence of
such disturbances motivates the use of the internal model
principle for disturbance rejection~\cite{Francis1976internal}, which states that regulation can be
achieved only if the feedback controller incorporates an augmented system that
is a copy of the exogenous system responsible for generating the sinusoidal
disturbances. This principle was thoroughly studied for linear systems in seminal works~\cite{Davison1976robust,Francis1976internal} and later 
generalized to address the nonlinear output regulation
problem~\cite{Isidori1990output,Serrani2001semi,BI2003tac,Huang2004general,Bayu2014}.
We refer to~\cite{Bin2022internal} for a comprehensive recent survey on this
subject. 
In recent years, the internal model principle has been employed to control
Euler--Lagrange systems subject to sinusoidal external disturbances. Several
linear internal model-based controllers were proposed
in~\cite{Chen1997adaptive,Bayu2008,Bayu2014,Wu2021CDC}, assuming prior
knowledge of the frequencies of external disturbances. To cope with unknown
frequencies in disturbances, adaptive internal model-based controllers were
proposed in~\cite{Lu2019,He2023output} and a nonlinear internal model-based
controller was developed in~\cite{Wu2022tac} for online estimation of the unknown
frequencies.  

A common feature of the aforementioned literature is that the design of the controller usually assumes the availability of velocity (or velocity error) measurements. In practice, modern position sensors such as encoders and cameras are able to provide low-noise and high-accuracy measurements of incremental joint angles and end-effector displacements, respectively. In contrast, obtaining velocity measurements, either directly or through numerical differentiation, increases system cost and is often prone to significant noise contamination. To mitigate the impact of noise in velocity measurements, there are mainly two classes of approaches. The first one is centered around making
compensations to counteract the effect of velocity noise, assuming that the
noise can be modeled, for instance, as harmonic signals as explored
in~\cite{Byrnes2003internal}. A notable result in this direction is the controller proposed in~\cite{Zarikian2007harmonic}, where two groups of internal models are
introduced for the Euler--Lagrange system to compensate for harmonic
disturbances present in the input and the velocity measurements, respectively.
The second approach aims to circumvent the use of velocity measurements and
instead focuses on developing controllers by employing velocity observers
(e.g.,~\cite{Andrieu2009unifying}) or filters
(e.g.,~\cite{Berghuis1993global,Kelly1993simple}). The present study
specifically concentrates on the latter approach. 

To eliminate the need for joint velocity measurements in robotic manipulators, considerable efforts have been devoted to velocity estimation. Among the various approaches, particular 
attention has been given to the globally convergent observer proposed
in~\cite{Besancon2000Auto}. This observer relies on the construction of a global
change of coordinates to transform the Euler--Lagrange equation into a
(partially) linearized form. However, finding such transformations for general
Euler--Lagrange systems remains challenging, and as of now, these
transformations are only known to exist for limited cases, as demonstrated
in~\cite{Besancon2000Auto} for 1-DOF systems and in~\cite{Yang2017TAC} for
two-link revolute robot manipulators. In general, global output feedback control of Euler--Lagrange systems is challenging, although it is a
subclass of strict-feedback systems,
see~\cite{Mazenc1994global,Andrieu2009unifying} for pioneering studies toward
global output feedback control of strict-feedback nonlinear systems. 
With regard to Euler--Lagrange systems specifically, the
independent papers~\cite{Berghuis1993global} and~\cite{Kelly1993simple} are
pioneering works that first proposed filter-type linear dynamic compensators to
solve this open problem. Since then, this method has been extensively used
both in practice and in the literature, see,
e.g.,~\cite{Ortega1995passivity,Loria1999force,Dirksz2012tracking,He2021leader,Li2023passivity}. 

A primary research interest of this study is to investigate task-space regulation and disturbance rejection for robot manipulators without using velocity measurements. We shall develop a filter-based control approach that integrates the internal model principle and passivity-based output-feedback techniques. Unlike conventional internal model-based methods that suppress nonlinear terms using high-gain functions, the proposed approach directly exploits their structural properties. The resulting control law is sufficiently smooth and does not require prior knowledge of the bounds of external disturbances, thereby simplifying the selection of controller gain parameters. In summary, the main contribution of the present study is the development of an internal model-based velocity-free control law that ensures asymptotic convergence of the regulation error and complete rejection of external disturbances. 

The remainder of this paper is organized as follows.  Section \ref{sec2} introduces the kinematics and dynamics of the manipulator and formulates the problem.  Section \ref{sec3} presents a full-state feedback disturbance rejection controller as a motivational design.  Section \ref{sec4} presents the main result of the present study. Simulation results are given in Section \ref{sec-sim}.  Section \ref{sec-con} closes this technical note.
\textit{Notation:} $\R^{n}$ is $n$-dimensional Euclidean space. $\| \cdot\| $ is the Euclidean norm. $\sigma(A)$ denotes the spectrum of matrix $A$. For a matrix $B\in\R^{n\times m}$, $B^{\t}$ denotes its transpose.

\section{Preliminaries}\label{sec2}
\subsection{Mathematical model of robotic systems}

Let $x\in\R^{n}$ be the task-space (e.g.,~Cartesian space) vector of a rigid
robot manipulator, and it is described as a nonlinear function of the joint
variable $q\in\R^{n}$ as follows~\citep{Cheah2008task}
\begin{equation}\label{sys-kine}
x = f(q).
\end{equation}
Taking its time derivative gives the velocity kinematics
\begin{equation}\label{sys-J}
\dot{x} = J(q)\dot{q},~~ J(q) := \frac{\partial f(q)}{\partial q} 
\end{equation}
where $\dot{q}\in\R^{n}$ is the joint velocity vector, and $J(q)\in\R^{m\times n}$ is the Jacobian matrix. The dynamic equation of a rigid $n$-link manipulator is given by
\begin{align}\label{sys-dyna}
H(q)\ddot{q} + C(q,\dot{q})\dot{q}  + g(q) = u + d
\end{align}
where $u\in\R^{n}$ is the joint control torque vector, $d\in\R^{n}$ is the external disturbance, $H(q)\in\R^{n \times n}$ is the inertia matrix, $C(q,\dot{q})\in\R^{n \times n}$ is the Coriolis and centrifugal force matrix-valued function, and $g(q)\in\R^{n}$ is the gravitational torque.

As in~\cite{Murray1994}, we list two properties of robot dynamics
\eqref{sys-dyna}: \\
\textit{Property~1.} $H(q)$ is uniformly positive definite; and \\
\textit{Property~2.} $\dot{H}(q,\dot{q})-2C(q,\dot{q})$ is skew-symmetric where $\dot{H}(q,\dot{q}) = \sum_{i=1}^{n}\frac{\partial H}{\partial q_{i}}\dot{q}_{i}$.

In the present study, we consider disturbance $d$ generated by the following system
\begin{equation}\label{exo00}
    \dot{w} = s(w),\quad d = \varphi(w,q)
\end{equation}
where $w$ is the state of the exosystem with appropriate dimension. As will be detailed later in Remark~\ref{rem00}, system \eqref{exo00} can describe input disturbance torques as well as external forces at end-effector.

\subsection{Problem formulation}
Combining \eqref{sys-kine}, \eqref{sys-dyna} and \eqref{exo00}, we can write the composite system in state-space representation
\begin{subequations}\label{sys-ss}
\begin{align}
\dot{w} &= s(w) \\
\dot{q} &= \xi \\
H(q)\dot{\xi} &=  - C(q,\xi)\xi - g(q) + u + \varphi(w,q) \\
e &= f(q) - x_{d}
\end{align}
\end{subequations}
where $\xi := \dot{q}$, $x_{d}$ is the constant desired position, and $e$ is the regulated output.

In this paper, two classes of feedback control schemes will be considered, namely full-state feedback control and velocity-free control. The former serves as a foundational design, laying the groundwork for the later main contribution. 

\begin{prob}\label{prob}
Consider the composite system \eqref{sys-ss}. 
\begin{itemize}
    \item [{\bf Q1}] If the joint velocity $\xi$ is available, design a smooth controller of the form
    \begin{equation}\label{pb:law00}
        u = h_{c}(x_{c},e,q,\xi),\quad \dot{x}_{c} = f_{c}(x_{c},e,q,\xi),
    \end{equation}
    or,
    \item [{\bf Q2}] if only the relative end-effector position and the joint position measurements, namely $(e,q)$, are available, design a smooth controller of the form
    \begin{equation}\label{pb:law01}
        u = h_{c}(x_{c},e,q),\quad \dot{x}_{c} = f_{c}(x_{c},e,q),
     \end{equation}
\end{itemize}
where $x_{c}$ is the state of dynamic compensator of appropriate dimension, such that $e(t)\to 0$ and $\xi(t)\to 0$ as $t\to\infty$.
\end{prob}

We investigate the problem under the following assumption. 

\begin{assum}\label{ass00} 
For the exosystem \eqref{exo00}, we assume that: \\
i) the exosystem is linear, i.e., 
\begin{equation}\label{exo01}
    s(w) = Sw
\end{equation}
for some matrix $S\in\R^{p\times p}$ whose eigenvalues are distinct and lie on the imaginary axis, and $S$ is nonsingular; and \\
ii) $\varphi(w,q)$ can be decomposed as
\begin{equation}\label{varphi}
    \varphi(w,q) = D_{1}w + J^{\t}(q)D_{2}w
\end{equation}
where $D_{1}$ and $D_{2}$ are matrices of appropriate dimension. 
\end{assum}

\begin{rem}\label{rem00} 
By \eqref{varphi}, the external disturbances under consideration can be decomposed as
\begin{equation}
\label{defn-d}
d = d_{1} + J^{\t}(q)d_{2}\quad \text{with}\quad d_{1} = D_{1}w,\ d_{2} = D_{2}w
\end{equation}
where $d_{1}\in\R^{n}$ represents the input disturbance torques, typically stemming from actuators, and $d_{2}\in\R^{n}$ represents the external forces acting on the end-effector. 

It should be noted that the system \eqref{exo01} in \ref{ass00} is a linear harmonic oscillator. This implies that each component of $d_{i}$, for $i=1,2$, is a combination of a finite number of sinusoidal signals, i.e.,~for $j=1,\dots,n$, $d_{ij}(t) = \sum_{k=1}^{N_{ij}}F_{ijk}\sin(\sigma_{ijk}t + \Upsilon_{ijk}) + F_{ij0}$ for some $N_{ij}\geq 0$, where $F_{ijk}$ and $\Upsilon_{ijk}$ are unknown parameters determined by the unknown initial condition $w(0)$, and the frequencies $\sigma_{ijk}$ are taken from the set of the eigenvalues of $S$.

In many electro-mechanical control systems, disturbances often appear as harmonic or periodic signals, primarily due to the dynamic behavior of rotational elements such as electric
 motors, gearboxes, and mechanical systems with vibrations (see,
 e.g.,~\cite{Zarikian2007harmonic,Tomizuka2008dealing}). It is well
 known from Fourier analysis that any continuous bounded periodic signal can
 be approximated by a truncated Fourier series. 
Therefore, if the disturbance signals $d_{1}$ and $d_{2}$ are periodic and can be expanded or approximated by a combination of a finite number of sinusoidal signals, they satisfy Assumption \ref{ass00}. 
In the literature, the rejection of such sinusoidal disturbances has been
 widely studied in the control of Euler--Lagrange
 systems~\cite{Lu2019,Wu2022tac} and in consensus problems involving multiple
 Euler--Lagrange systems~\cite{Wang2023leaderless}. 
\end{rem}

\subsection{Passivity of transpose Jacobian feedback control}
An important property of the robot manipulator is that it exhibits passivity properties in both joint and task spaces.\footnote{For the definition of passivity, we refer to \cite[Chapter~6]{Khalil2002}.}

\begin{prop}\label{prop0}
Consider system \eqref{sys-kine}, \eqref{sys-J} and \eqref{sys-dyna} with Properties 1 and 2. Suppose that  $d=0$. Then the following properties hold:
    \begin{itemize}[topsep=0pt,parsep=0pt,partopsep=0pt]
        
       \item [{\bf P1}]  System \eqref{sys-kine}, \eqref{sys-J} and \eqref{sys-dyna} with control input
        \begin{align}\label{prop:law00}
            u = - kJ^{\t}(q)x + v + g(q),\quad k>0
        \end{align}
        is passive with input $v$ and output $\dot{q}$.
        
        \item [{\bf P2}] System \eqref{sys-kine}, \eqref{sys-J} and \eqref{sys-dyna} with control input
        \begin{align}\label{prop:law01}
            u = - kJ^{\t}(q)x + J^{\t}(q)F + g(q),\quad k>0
        \end{align}
        is passive with input $F$ and output $\dot{x}$.

    \end{itemize}
\end{prop}

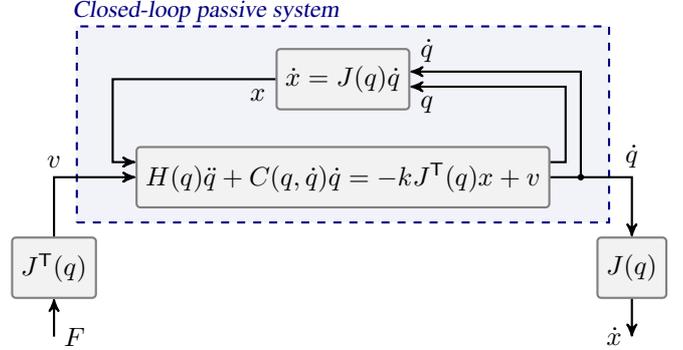
\begin{figure}[htbp]
    \centering
\begin{tikzpicture}[auto,node distance=1.0cm, >=stealth',scale=1]

\draw [thick, dashed, draw=NavyBlue, fill=NavyBlue!5, rounded corners=0mm] (-3.5cm,-0.6cm) rectangle (3.5cm,2.0cm); 
\draw []  (-1.8cm,2.2cm) node{\color{NavyBlue}\small\sl Closed-loop passive system};

\node[block] (plant) {$H(q)\ddot{q}+C(q,\dot{q})\dot{q}=-kJ^{\t}(q)x + v$}; 
\node[block, above of=plant, node distance=1.3cm] (x) {$\dot{x}=J(q)\dot{q}$};
\node[block, below of=plant, node distance=1.2cm, xshift=3.8cm] (J) {$J(q)$};
\node[block, below of=plant, node distance=1.2cm, xshift=-3.8cm] (JT) {$J^{\t}(q)$};


\draw [->, thick] (plant) -| node[above]{$\dot{q}$} (J);

\draw [->, thick] ($(plant.east)+(0cm,0.2cm)$) -- ($(plant.east)+(0.2cm,0.2cm)$) |- ($(x.east)+(0cm,-0.1cm)$) node[below right]{$q$};
\draw [->, thick] ($(plant.east)+(0.4cm,0cm)$) |- ($(x.east)+(0cm,0.1cm)$) node[above right]{$\dot{q}$};
\filldraw[] ($(plant.east)+(0.4cm,0cm)$) circle (1pt);
\draw [->, thick] ($(x.west)+(0cm,0cm)$) node[below left]{$x$} -|  ($(plant.west)+(-0.3cm,0.2cm)$) -- ($(plant.west)+(0cm,0.2cm)$);

\draw [->, thick] (JT) |- node[above]{$v$} (plant);
\draw [->, thick] (J) -- ($(J.south)+(0cm,-0.5cm)$) node[left]{$\dot{x}$};
\draw [->, thick] ($(JT.south)+(0cm,-0.5cm)$) node[right]{$F$} -- (JT);

\end{tikzpicture} 
\caption{Passivity interpretation of the controllers in Proposition~\ref{prop0}. }
\label{fig1-001}
\end{figure}

\begin{proof}
Define a storage function $V_{1} = \frac{1}{2}x^{\t}kx + \frac{1}{2}\dot{q}^{\t}H(q)\dot{q}$. Differentiation along the trajectories of \eqref{sys-J}, \eqref{sys-dyna}, \eqref{prop:law00} yields $\dot{V}_{1} = \dot{q}^{\t}v$, which implies that the system is passive with input $v$ and output $\dot{q}$. 

Further differentiating $V_{1}$ along the trajectories of \eqref{sys-J}, \eqref{sys-dyna}, \eqref{prop:law01} yields $\dot{V}_{1} = \dot{q}^{\t}J^{\t}(q)F = [J(q)\dot{q}]^{\t}F = \dot{x}^{\t}F$, which implies that the system is passive with input $F$ and output $\dot{x}$. 
\end{proof}

The proof shows that the system \eqref{sys-J}, \eqref{sys-dyna} and \eqref{prop:law00} (or \eqref{prop:law01}) is also lossless \cite{Khalil2002}. Figure~\ref{fig1-001} illustrates the closed-loop passive (lossless) mappings. Note that the effect of the external disturbances is not given in Proposition~\ref{prop0}. In the following, we will use these passivity properties to study interactions with external disturbances as well as internal model-based dynamic compensators.

\section{Full-state feedback control}\label{sec3}
This section introduces a solution to problem {\bf Q1} as the first design step. Subsequently, in the next section, we will focus on addressing problem {\bf Q2}. 
The methodology used in this section follows closely the one in \cite{Bayu2008} and uses the passivity properties described in the previous section. In particular, inspired by the internal model principle, a pair of internal model candidates are employed to make compensation for the two external disturbances in \eqref{defn-d}, respectively.

\begin{itemize}[topsep=0pt,parsep=0pt,partopsep=0pt]
    \item[1)] To counteract the effect of $d_{1}$, we introduce an internal model of the following form 
    \begin{equation}\label{im01}
    \dot{\eta}_{1} = A_{1}\eta_{1} - B_{1}\xi,\quad 
    \hat{d}_{1} = B_{1}^{\t}\eta_{1}
    \end{equation}
    with state $\eta_{1}\in\R^{\ell_{1}}$. 

    \item[2)] To counteract the effect of $d_{2}$ at the end-effector, we introduce an internal model of the following form
    \begin{equation}\label{im02}
    \dot{\eta}_{2} = A_{2}\eta_{2} - B_{2}J(q)\xi,\quad 
    \hat{d}_{2} = B_{2}^{\t}\eta_{2}
    \end{equation}
    with state $\eta_{2}\in\R^{\ell_{2}}$. Since $\dot{x} = J(q)\xi$, internal model \eqref{im02} can be written as $\dot{\eta}_{2} = A_{2}\eta_{2} - B_{2}\dot{x}$, which is driven by $\dot{x}$. 
\end{itemize}

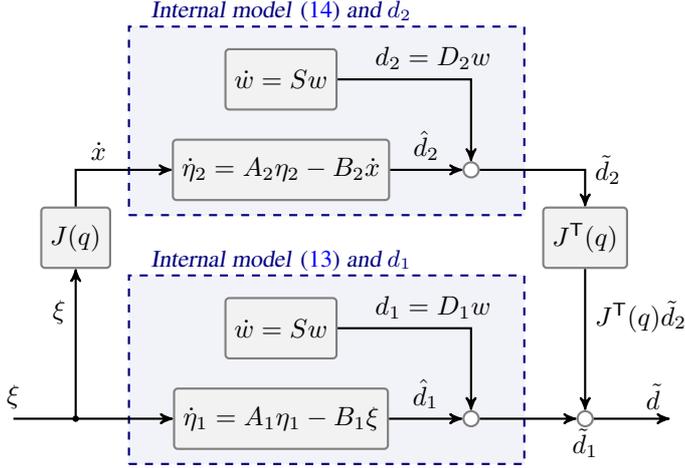
\begin{figure}[htbp]
    \centering
\begin{tikzpicture}[auto,node distance=1.0cm, >=stealth',scale=1]

\draw [thick, dashed, draw=NavyBlue, fill=NavyBlue!5, rounded corners=0mm] (-2.0cm,-1.8cm) rectangle (3.2cm,0.7cm); 
\draw []  (0cm,0.9cm) node{\color{NavyBlue}\small\sl Internal model \eqref{im01} and $d_{1}$}; 

\node[block] (exo1) {$\dot{w} = Sw$}; 
\node[block, below of=exo1, node distance=1.2cm] (im1) {$\dot{\eta}_{1} = A_{1}\eta_{1} - B_{1}\xi$};
\node[sum, right of=im1, node distance=2.5cm] (sum1) {};
\node[sum, right of=sum1, node distance=1.5cm] (sum11) {};

\draw [thick, dashed, draw=NavyBlue, fill=NavyBlue!5, rounded corners=0mm] (-2.0cm,1.5cm) rectangle (3.2cm,4.0cm);
\draw []  (0cm,4.2cm) node{\color{NavyBlue}\small\sl Internal model \eqref{im02} and $d_{2}$}; 
\node[block, above of=im1, node distance=4.5cm] (exo2) {$\dot{w} = Sw$}; 
\node[block, below of=exo2, node distance=1.2cm] (im2) {$\dot{\eta}_{2} = A_{2}\eta_{2} - B_{2}\dot{x}$};
\node[sum, right of=im2, node distance=2.5cm] (sum2) {};
\node[block, right of=sum2, node distance=1.5cm, yshift=-0.9cm] (JT) {$J^{\t}(q)$};
\node[block, below of=im2, node distance=0.9cm, xshift=-2.7cm] (J) {$J(q)$};


\draw [->, thick] (im1) -- node[above]{$\hat{d}_{1}$} (sum1);
\draw [->, thick] (sum1) -- (sum11) node[below]{$\tilde{d}_{1}$};
\draw [->, thick] (sum11) -- ($(sum11.east)+(1.0cm,0cm)$) node[above left]{$\tilde{d}$};
\draw [->, thick] ($(im1.west)+(-2.1cm,0cm)$) node[above]{$\xi$} --  (im1);
\draw [->, thick] (exo1)  -| node[above, xshift=-0.5cm]{$d_{1}=D_{1}w$} (sum1);

\draw [->, thick] (im2) -- node[above]{$\hat{d}_{2}$} (sum2);
\draw [->, thick] (sum2) -| node[right]{$\tilde{d}_{2}$} (JT);
\draw [->, thick] (J) |-  node[above, xshift=0.3cm]{$\dot{x}$} (im2);
\draw [->, thick] (exo2)  -| node[above, xshift=-0.5cm]{$d_{2}=D_{2}w$} (sum2);
\draw [->, thick] ($(im1)+(-2.7cm,-0.0cm)$) -- (J) node[below left, yshift=-0.7cm]{$\xi$} ;
\draw [->, thick] (JT) node[below right, yshift=-0.7cm]{$J^{\t}(q)\tilde{d}_{2}$} --   (sum11) ;
\filldraw[] ($(im1)+(-2.7cm,-0.0cm)$) circle (1pt);

\end{tikzpicture} 
\caption{Modules of exosystem and internal models \eqref{im01} and \eqref{im02} in closed-loop system block diagram.} 
\label{fig1-002}
\end{figure}

In \eqref{im01} and \eqref{im02}, we designed two internal models for counteracting the external disturbances. By interconnecting these internal models and the exosystem appropriately as in Figure~\ref{fig1-002}, the resulting system has a passivity (lossless) property if the design parameters $(A_{i},B_{i})$ for $i=1,2$ are chosen such that the following condition holds.

\begin{assum}\label{assum01}
    For each $i=1,2$, there exists a matrix $\Sigma_{i}\in\R^{\ell_{i}\times p}$ such that
    \begin{equation}\label{Sigma-i}
        \Sigma_{i}S = A_{i}\Sigma_{i},\quad  
        B_{i}^{\t}\Sigma_{i} + D_{i} = 0. 
    \end{equation}
    Moreover, $(A_{i},B_{i}^{\t})$ is observable, and $A_{i}$ is skew-symmetric and nonsingular.
\end{assum}

Define error variables $\tilde{\eta}_{i} = \eta_{i} - \Sigma_{i}w$, $\tilde{d}_{i} = d_{i} + \hat{d}_{i}$ for $i=1,2$, and $\tilde{d} = \tilde{d}_{1} + J^{\t}(q)\tilde{d}_{2}$. Then, under Assumption \ref{assum01}, 
\begin{align}\label{im01-e}
    \dot{\tilde{\eta}}_{1} = A_{1}\tilde{\eta}_{1} - B_{1}\xi,\quad  
    \dot{\tilde{\eta}}_{2} = A_{2}\tilde{\eta}_{2} - B_{2}J(q)\xi
\end{align}
and $\tilde{d} = B_{1}^{\t}\tilde{\eta}_{1} + J^{\t}(q)B_{2}^{\t}\tilde{\eta}_{2}$. 
With the storage function $V_{2} = \frac{1}{2}\tilde{\eta}_{1}^{\t}\tilde{\eta}_{1} + \frac{1}{2}\tilde{\eta}_{2}^{\t}\tilde{\eta}_{2}$, it can be verified that the error system \eqref{im01-e} is lossless with input $\xi$ and output $\tilde{d}$.

Based on the aforementioned passivity analysis, we propose a full-state feedback controller for {\bf Q1} in the following proposition.

\begin{prop}\label{thm00}
Consider the system \eqref{sys-ss} under Assumptions \ref{ass00} and \ref{assum01}, and feedback-interconnected with the controller
\begin{subequations}\label{law-00}
\begin{align}
        \dot{\eta}_{1} &= A_{1}\eta_{1} - B_{1}\xi \\
        \dot{\eta}_{2} &= A_{2}\eta_{2} - B_{2}J(q)\xi \\
        u &= - k_{p}J^{\t}(q)e - k_{d}\xi + g(q)  + B_{1}^{\t}\eta_{1} + J^{\t}(q)B_{2}^{\t}\eta_{2}
\end{align}
\end{subequations}
where $k_{p}, k_{d}>0$. Then, for a finite task space in which the Jacobian matrix $J(q)$ has full rank, the regulation error and velocity asymptotically converges to zero as time $t\to\infty$, i.e.,~$\lim_{t\to\infty} e(t) = 0$ and $\lim_{t\to\infty} \xi(t) = 0$. 

\end{prop}

\begin{proof}
By Assumption \ref{assum01}, there exist matrices $\Sigma_{1}$ and $\Sigma_{2}$ such that the equations in \eqref{Sigma-i} hold. 
Denote
\begin{align}\label{eq:A}
&D = 
\begin{bmatrix}
 D_{1}  \\  D_{2} 
\end{bmatrix}
,\ 
\Sigma = 
\begin{bmatrix}
 \Sigma_{1}  \\  \Sigma_{2} 
\end{bmatrix}
,\
\Gamma(q) = 
\begin{bmatrix}
 I \\ J(q) 
\end{bmatrix}
 \nonumber\\
&A = 
\begin{bmatrix}
 A_{1} & 0 \\ 0 & A_{2} 
\end{bmatrix}
,\ 
B = 
\begin{bmatrix}
 B_{1} & 0 \\ 0 & B_{2} 
\end{bmatrix}
,\
\eta = 
\begin{bmatrix}
 \eta_{1} \\ \eta_{2} 
\end{bmatrix}
.
\end{align}
Using \eqref{eq:A}, the closed-loop system \eqref{sys-ss} and \eqref{law-00} under the coordinate transformation $\bar{\eta} = \eta - \Sigma w$ can be written as
\begin{subequations}
\label{cls01}
\begin{align}
        \dot{\bar{\eta}} &= A\bar{\eta} - B\Gamma(q)\xi \\
        \dot{q} &= \xi \\
        H(q)\dot{\xi} &= -k_{p}J^{\t}(q)e - k_{d}\xi -C(q,\xi)\xi + \Gamma^{\t}(q)B^{\t}\bar{\eta} 
\end{align}
\end{subequations}
in which $\Gamma^{\t}(q)B^{\t}\Sigma w + \Gamma^{\t}(q)D w$ has been canceled by \eqref{Sigma-i}. 
    
Define a Lyapunov function candidate $V\coloneq  V(\bar{\eta},q,\xi)$ by $V = \frac{1}{2}(f(q) - x_{d})^{\t}k_{p}(f(q) - x_{d}) + \frac{1}{2}\xi^{\t}H(q)\xi + \frac{1}{2}\bar{\eta}^{\t}\bar{\eta}$ whose time derivative, along the trajectories of \eqref{cls01}, satisfies
\begin{align*}
        \dot{V} &= e^{\t}k_{p}J(q)\xi + \xi^{\t}\bigl[ -k_{p}J^{\t}(q)e - k_{d}\xi + \Gamma^{\t}(q)B^{\t}\bar{\eta} \bigr] \\ 
        &\quad + \bar{\eta}^{\t}(A\bar{\eta} - B\Gamma(q)\xi)  = -\xi^{\t}k_{d}\xi.
\end{align*}
Since $\dot{V} \leq 0$ and $V\geq 0$, $V$ is bounded for all $t\geq 0$. Hence, $(\bar{\eta}(t),e(t),\xi(t))$ are all bounded over the time interval $[0,\infty)$.

In the following, we will apply LaSalle's invariance theorem~\cite[Theorem~4.4]{Khalil2002} to establish the asymptotic stability of the invariant set.
To this end, we need to find the largest invariant set in $\{(\bar{\eta},q,\xi): \dot{V} = 0 \} = \{(\bar{\eta},q,\xi): \xi = 0 \}$. Substituting $\xi = 0$ into \eqref{cls01} gives
    \begin{subequations}\label{cls02}
    \begin{align}
        \dot{\bar{\eta}} &= A\bar{\eta}  \label{cls02-a}\\
        \dot{q} &= 0  \label{cls02-b}\\ 
        0 &= -k_{p}J^{\t}(q)e  + \Gamma^{\t}(q)B^{\t}\bar{\eta}.  \label{cls02-c}
    \end{align}
    \end{subequations}

    Differentiating \eqref{cls02-c} with respect to $t$, $\ell$ times, where $\ell=\ell_{1}+\ell_{2}$, and using $\dot{J}(q) = \sum_{i=1}^{n}\frac{\partial J}{\partial q_{i}}\xi_{i} = 0$ and $\dot{e}=J(q)\xi = 0$ when $\xi = 0$, we obtain
    \begin{equation*}
        \left\{\begin{array}{l}        
        0 = \Gamma^{\t}(q)B^{\t}A\bar{\eta} \\
        \quad \vdots \\
        0 = \Gamma^{\t}(q)B^{\t}A^{\ell-1}\bar{\eta}  \\
        0 = \Gamma^{\t}(q)B^{\t}A^{\ell}\bar{\eta}\end{array}\right.
        \ \Rightarrow\
        \left\{\begin{array}{l}        
        0 = \Gamma^{\t}(q)B^{\t}(\alpha_{l-1}A)\bar{\eta} \\
        \quad \vdots \\
        0 = \Gamma^{\t}(q)B^{\t}(\alpha_{1} A^{\ell-1})\bar{\eta}  \\
        0 = \Gamma^{\t}(q)B^{\t} A^{\ell}\bar{\eta}\end{array}\right.
\end{equation*}
    in which $\alpha_{1},\dots,\alpha_{l}$ are real numbers such that (by the Cayley--Hamilton theorem) $A^{l} + \alpha_{1}A^{l-1} +\cdots+ \alpha_{l-1}A + \alpha_{l}I = 0$ and $\alpha_{l} \neq 0$. 
    Hence, $
        \alpha_{l}\Gamma^{\t}(q)B^{\t}\bar{\eta} =
        \Gamma^{\t}(q)B^{\t}\bigl(\alpha_{l} I\bigr)\bar{\eta} 
        = \Gamma^{\t}(q)B^{\t}(-\alpha_{l-1}A -\cdots- \alpha_{1} A^{\ell-1} - A^{\ell} )\bar{\eta} 
        = 0$. It follows that $\Gamma^{\t}(q)B^{\t}\bar{\eta} = 0$ holds in the invariant set. Substituting $\Gamma^{\t}(q)B^{\t}\bar{\eta} = 0$ into \eqref{cls02} results in $0 = -k_{p}J^{\t}(q)e$. 
    This means that $e=0$ as long as $J(q)$ is full rank. Hence the largest invariant set in $\{(\bar{\eta},q,\xi): \dot{V} = 0 \}$ w.r.t. \eqref{cls01} is $\Omega \coloneq  \{(\bar{\eta},q,\xi) : \Gamma^{\t}(q)B^{\t}\bar{\eta}=0, e=0, \xi=0\}$.
    
    Finally, by LaSalle's invariance principle, we can conclude that the state trajectories $(e,\xi)$ asymptotically converge to $(0,0)$ as $t\to\infty$.    
\end{proof}

Notice that the proposed controller \eqref{law-00} is based on the internal model principle and specifically incorporates a pair of parallel internal models. Assumption \ref{assum01} asks for the existence of matrices $\Sigma_{1}$ and $\Sigma_{2}$ for the condition~\eqref{Sigma-i} to hold. This condition is assumed separately for $i=1,2$ and is generally not necessary for achieving asymptotic regulation with disturbance rejection. In particular, if $A_{1}$ and $A_{2}$ have the same spectrum, a combined form of this condition for $i=1,2$ may also be effective. However, if $A_{1}$ and $A_{2}$ have no common eigenvalues and the two internal models are designed to counteract the effects of $d_{1}$ and $d_{2}$, respectively, the condition~\eqref{Sigma-i} becomes necessary. 
To demonstrate this necessity property, let us focus on the invariant set $\Omega'=\{(\eta,q,\xi):e=0,\xi=0\}$, where $e=f(q)-x_{d}$, within the closed-loop system consisting of \eqref{sys-ss} and \ref{law-00}.
On the invariant set $\Omega'$, 
\begin{subequations}\label{cls03}
\begin{align}
    \dot{\eta} &= A\eta   \label{cls03-a}\\ 
    \dot{q} &= 0   \label{cls03-b}\\
    0 &= \Gamma^{\t}(q)B^{\t}\eta + \Gamma^{\t}(q)Dw.  \label{cls03-c}
\end{align}
\end{subequations}
 Differentiating \eqref{cls03-c} with respect to $t$, $\ell-1$ times, gives
\begin{equation}\label{eta-w}
    \Phi_{1}\eta + \Phi_{2}w = 0
\end{equation}
where $\Phi_{1} = \bbm{ \Gamma^{\t}(q)B^{\t}  \\ \vdots \\ \Gamma^{\t}(q)B^{\t}A^{\ell - 1} } $ and $\Phi_{2} = \bbm{ \Gamma^{\t}(q)D \\ \vdots \\ \Gamma^{\t}(q)DS^{\ell - 1} }$.

To proceed, we use an intermediate technical lemma whose proof is put in Appendix. 

\begin{lem}\label{lem-app}
Consider $n_{i}$-dimensional $q_{i}$-output observable pairs $(\bfA_{i},\bfC_{i})$, $i=1,2$. Let $\bfT$ be a $q_{1}\times q_{2}$ matrix of rank $q_{2}$.   Denote
\begin{align*}
    \bfA = \begin{bmatrix} \bfA_{1} & 0 \\ 0 & \bfA_{2} \end{bmatrix},\
    \bfC = \begin{bmatrix} \bfC_{1} & \bfT \bfC_{2} \end{bmatrix}.
\end{align*}
If $\sigma(\bfA_{1})\cap\sigma(\bfA_{2})=\emptyset$ then $(\bfA,\bfC)$ is observable.
\end{lem}

Suppose that $\sigma(A_{1})\cap\sigma(A_{2})=\emptyset$. Since $(A_{1},B_{1}^{\t})$ and $(A_{2},B_{2}^{\t})$ are observable, and $J(q)$ has full rank, we can conclude from Lemma~\ref{lem-app} that $(A,\Gamma^{\t}(q)B^{\t})$ is observable. 
Consequently, $\Phi_{1}$ has full rank and $\Phi_{1}^{\t}\Phi_{1}$ is invertible. 
Combining this with \eqref{eta-w}, we obtain
\begin{align}\label{Sigma03}
    \eta(t) = \Sigma w(t) \ \ \text{where}\ \ \Sigma = -\left(\Phi_{1}^{\t}\Phi_{1}\right)^{-1}\Phi_{1}^{\t}\Phi_{2}.
\end{align}
On the set $\Omega'$, we have the following chain of implications
\begin{align}\label{cdtion01}
    &\dot{\eta} \stackrel{\eqref{cls03-a}}{=} A\eta \stackrel{\mathclap{\eqref{Sigma03}}}{=} A\Sigma w \quad \text{and}\quad 
    \Sigma\dot{w} \stackrel{\eqref{exo00},\eqref{exo01}}{=} \Sigma Sw \nonumber\\ 
    &\stackrel{\mathclap{\dot{\eta} = \Sigma\dot{w}}}{\Rightarrow} \ \ 
    A\Sigma w = \Sigma Sw \ 
    \stackrel{\mathclap{\eqref{eq:A}}}{\Rightarrow} 
    A_{i}\Sigma_{i}w = \Sigma_{i}Sw,\ i=1,2.
\end{align}
By substituting $\eta = \Sigma w$ into \eqref{cls03-c}, we obtain $\Gamma^{\t}(q)B^{\t}\Sigma w + \Gamma^{\t}(q)Dw = 0$. This implies that if the spectrum of $A_{1}$ is different to that of $A_{2}$ and the two internal models are designed to counteract the effects of $d_{1}$ and $d_{2}$ respectively, then
\begin{equation}\label{cdtion02}
    B_{i}^{\t}\Sigma_{i} w + D_{i}w = 0,\quad i=1,2. 
\end{equation}
Since \eqref{cdtion01} and \eqref{cdtion02} hold for all $t\geq 0$ and $w(0)$ is such that all modes of the exosystem are excited, this implies that the equations in \eqref{Sigma-i} hold for the closed-loop system \eqref{sys-ss} and \eqref{law-00} on the invariant set $\Omega'$.

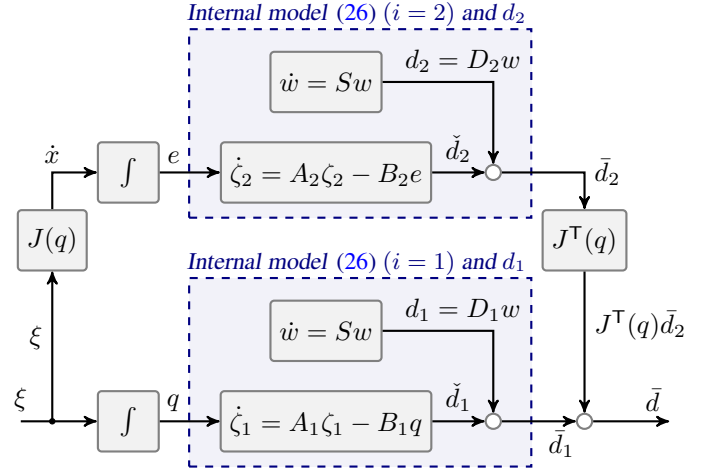
\begin{figure}[htbp]
    \centering
\begin{tikzpicture}[auto,node distance=1.0cm, >=stealth',scale=1]

\draw [thick, dashed, draw=NavyBlue, fill=NavyBlue!5, rounded corners=0mm] (-1.8cm,-1.8cm) rectangle (2.7cm,0.7cm); 
\draw []  (0.4cm,0.9cm) node{\color{NavyBlue}\small\sl Internal model \eqref{im10} $(i=1)$ and $d_{1}$}; 

\node[block] (exo1) {$\dot{w} = Sw$}; 
\node[block, below of=exo1, node distance=1.2cm] (im1) {$\dot{\zeta}_{1} = A_{1}\zeta_{1} - B_{1}q$};
\node[block, left of=im1, node distance=2.6cm] (int1) {$\int$};
\node[sum, right of=im1, node distance=2.2cm] (sum1) {};
\node[sum, right of=sum1, node distance=1.2cm] (sum11) {};

\draw [thick, dashed, draw=NavyBlue, fill=NavyBlue!5, rounded corners=0mm] (-1.8cm,1.5cm) rectangle (2.7cm,4.0cm);
\draw []  (0.4cm,4.2cm) node{\color{NavyBlue}\small\sl Internal model \eqref{im10} $(i=2)$ and $d_{2}$}; 
\node[block, above of=im1, node distance=4.5cm] (exo2) {$\dot{w} = Sw$}; 
\node[block, below of=exo2, node distance=1.2cm] (im2) {$\dot{\zeta}_{2} = A_{2}\zeta_{2} - B_{2}e$};
\node[block, left of=im2, node distance=2.6cm] (int2) {$\int$};
\node[sum, right of=im2, node distance=2.2cm] (sum2) {};
\node[block, right of=sum2, node distance=1.2cm, yshift=-0.9cm] (JT) {$J^{\t}(q)$};
\node[block, below of=int2, node distance=0.9cm, xshift=-1.0cm] (J) {$J(q)$};

\draw [->, thick] (im1) -- node[above]{$\check{d}_{1}$} (sum1);
\draw [->, thick] (sum1) -- (sum11) node[below left]{$\bar{d}_{1}$};
\draw [->, thick] (sum11) -- ($(sum11.east)+(1.0cm,0cm)$) node[above left]{$\bar{d}$};
\draw [->, thick] (int1) --node[above left]{$q$} (im1);
\draw [->, thick] ($(int1.west)+(-1.0cm,0cm)$) node[above]{$\xi$} --  (int1);
\draw [->, thick] (exo1)  -| node[above, xshift=-0.4cm]{$d_{1}=D_{1}w$} (sum1);

\draw [->, thick] (im2) -- node[above]{$\check{d}_{2}$} (sum2);
\draw [->, thick] (sum2) -| node[right]{$\bar{d}_{2}$} (JT);
\draw [->, thick] (J) |-  node[above]{$\dot{x}$} (int2);
\draw [->, thick] (int2) --node[above left]{$e$}  (im2);
\draw [->, thick] (exo2)  -| node[above, xshift=-0.4cm]{$d_{2}=D_{2}w$} (sum2);
\draw [->, thick] ($(int1)+(-1.0cm,-0.0cm)$) -- (J) node[below left, yshift=-1cm]{$\xi$} ;
\draw [->, thick] (JT) node[below right, yshift=-0.8cm]{$J^{\t}(q)\bar{d}_{2}$} --   (sum11) ;
\filldraw[] ($(int1)+(-1.0cm,-0.0cm)$) circle (1pt);

\end{tikzpicture} 
\caption{Modules of exosystem and internal model \eqref{im10} in closed-loop system block diagram. } 
\label{fig1-003}
\end{figure}

\section{Velocity-free control}\label{sec4}

This section is devoted to developing a velocity-free controller for solving {\bf Q2}. To avoid using velocity measurements in the internal model dynamics, we propose to modify internal models \eqref{im01} and \eqref{im02}, respectively, as
\begin{equation}\label{im10}
    \begin{aligned}
        &\dot{\zeta}_{i} = A_{i}\zeta_{i} - B_{i}y_{i},\quad  y_{1} = q,\ y_{2} = e \\
        &\check{d}_{i} = B_{i}^{\t}(A_{i}\zeta_{i} - B_{i}y_{i})
    \end{aligned}
\end{equation}
for $i=1,2$, with $\zeta_{i}\in\R^{\ell_{i}}$ and $A_{i},B_{i}$ as in \eqref{im01} and \eqref{im02}.

Similar to that in full-state feedback control, interconnecting internal models \eqref{im10} and the exosystem as in Figure~\ref{fig1-003} gives rise to a system having lossless property for appropriate design parameters. To this end, let us define error variables $\bar{\zeta}_{i} = A_{i}\zeta_{i} - B_{i}y_{i} - \Sigma_{i}w$, $\bar{d}_{i} = d_{i} + \check{d}_{i}$ for $i=1,2$, and $\bar{d} = \bar{d}_{1} + J^{\t}(q)\bar{d}_{2}$. Then, under Assumption \ref{assum01}, 
\begin{align}\label{im10-e}
    \dot{\bar{\zeta}}_{1} = A_{1}\bar{\zeta}_{1} - B_{1}\xi,\quad 
    \dot{\bar{\zeta}}_{2} = A_{2}\bar{\zeta}_{2} - B_{2}J(q)\xi  
\end{align}
and $\bar{d} = B_{1}^{\t}\bar{\zeta}_{1} + J^{\t}(q)B_{2}^{\t}\bar{\zeta}_{2}$. With the storage function $V_{3} = \frac{1}{2}\bar{\zeta}_{1}^{\t}\bar{\zeta}_{1} + \frac{1}{2}\bar{\zeta}_{2}^{\t}\bar{\zeta}_{2}$, it can be verified that the error system \eqref{im10-e} is lossless with input $\xi$ and output $\bar{d}$.

To eliminate the need for velocity measurements in the stabilization part, we
introduce a filter-type dynamic compensator that uses its output as a
substitute for joint velocity measurements. Although this approach, originating from the seminal works of~\cite{Berghuis1993global,Kelly1993simple}, has been widely applied in the control of Euler--Lagrange systems, it remains essential to
investigate whether the integration of the designed filter with the lossless
internal model-based disturbance compensator will maintain stability and ensure
the asymptotic convergence of the regulation error.
The main result of the present study is given as follows.

\begin{prop}\label{thm01}
Consider the system \eqref{sys-ss} under Assumptions \ref{ass00} and \ref{assum01}, and feedback-interconnected with the controller
\begin{subequations}\label{law-02}
\begin{align}
        \dot{\zeta}_{1} &= A_{1}\zeta_{1} - B_{1}q \\
        \dot{\zeta}_{2} &= A_{2}\zeta_{2} - B_{2}e \\
        \dot{\chi} &= -h(\chi + hq) \\
        u &= -k_{p}J^{\t}(q)e - k_{d}(\chi + hq) + g(q) \nonumber\\
& \quad  + B_{1}^{\t}(A_{1}\zeta_{1} - B_{1}q) +  J^{\t}(q)B_{2}^{\t}(A_{2}\zeta_{2} - B_{2}e)
\end{align}
\end{subequations}
where $k_{p}, k_{d}>0$. Then, for a finite task space in which the Jacobian matrix $J(q)$ has full rank, the regulation error and velocity asymptotically converges to zero as time $t\to\infty$, i.e.,~$\lim_{t\to\infty} e(t) = 0$, $\lim_{t\to\infty} \xi(t) = 0$.   

\end{prop}

\begin{proof}
By Assumption \ref{assum01}, there exist $\Sigma_{1}$ and $\Sigma_{2}$ satisfying \eqref{Sigma-i}. 
Define $\zeta = [\zeta_{1}^{\t}, \zeta_{2}^{\t}]^{\t}$. 
Applying the linear coordinate transformation $\bar{\zeta} = A\zeta - By - \Sigma w$ to the closed-loop system \eqref{sys-ss} and \eqref{law-02} gives
\begin{subequations}
\label{cls10}
\begin{align}
    \dot{\bar{\zeta}} &= A\bar{\zeta} - B\Gamma(q)\xi  \label{cls10-a}\\
    \dot{\chi} &= -h(\chi + hq)  \label{cls10-b}\\ 
    \dot{q} &= \xi  \label{cls10-c}\\
    H(q)\dot{\xi} &= -k_{p}J^{\t}(q)e - k_{d}(\chi + hq) \nonumber\\
    &\quad -C(q,\xi)\xi + \Gamma^{\t}(q)B^{\t}\bar{\zeta}.  \label{cls10-d}
\end{align}
\end{subequations}
Define $\bar{V}\coloneq \bar{V}(\bar{\zeta},\chi,q,\xi) = \frac{1}{2}\bar{\zeta}^{\t}\bar{\zeta} + \frac{1}{2}k_{d}(\chi + hq)^{\t}h^{-1}(\chi + hq) + \frac{1}{2}k_{p}(f(q) - x_{d})^{\t}(f(q) - x_{d})  + \frac{1}{2}\xi^{\t}H(q)\xi$. 
Its time derivative, along the trajectories of \eqref{cls10}, satisfies $\dot{\bar{V}} = -k_{d}(\chi + hq)^{\t}(\chi + hq)$. 
Since $\dot{\bar{V}} \leq 0$ and $\bar{V}\geq 0$, $\bar{V}$ is bounded for all $t\geq 0$. Hence, $(\bar{\zeta}(t),\dot{\chi}(t),e(t),\xi(t))$ are all bounded over the time interval $[0,\infty)$.

As before, LaSalle's invariance theorem can be applied to complete the proof. To find the largest invariant set in
\begin{equation}
\label{inv10}
\bigl\{(\bar{\zeta},\chi,q,\xi):\dot{\bar{V}}=0\bigr\} ~~\text{or}~~ \bigl\{(\bar{\zeta},\chi,q,\xi):\dot{\chi}=0\bigr\}
\end{equation}
we notice that $\dot{\chi} = 0$ implies $\chi$ is a constant vector, and hence by \eqref{cls10-b}, $q$ is also a constant vector. 
It follows that $\xi=0$.

Following similar reasoning as in the proof of Proposition \ref{thm00}, we can conclude that the largest invariant set in \eqref{inv10} w.r.t. \eqref{cls10} is the set $\bar{\Omega} \coloneq  \bigl\{(\bar{\zeta},\chi,q,\xi) : \Gamma^{\t}(q)B^{\t}\bar{\zeta}=0,\dot{\chi}=0,e=0,\xi=0\bigr\}$ 
in which $\dot{\chi} = -h(\chi + hq)$ and $e=f(q)-x_{d}$. Finally, by LaSalle's invariance principle, we conclude that $\lim_{t\to\infty}e(t) = 0$ and $\lim_{t\to\infty}\xi(t) = 0$. 
\end{proof}

\begin{figure*}[!ht]
     \centering
     \begin{subfigure}[t]{0.30\textwidth}
         \centering
         \includegraphics[width=\textwidth, trim={14 3 33 20}, clip]{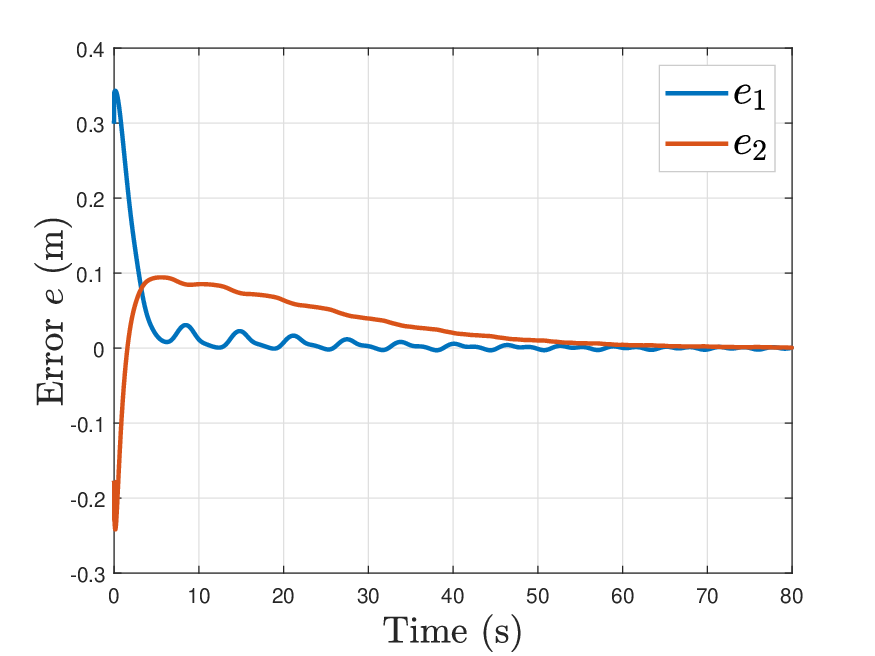}
         \caption{Error $e = [e_{1}, e_{2}]^{\t}$}
         \label{fig:e}
     \end{subfigure}
     \hfill
     \begin{subfigure}[t]{0.30\textwidth}
         \centering
         \includegraphics[width=\textwidth, trim={14 3 33 20}, clip]{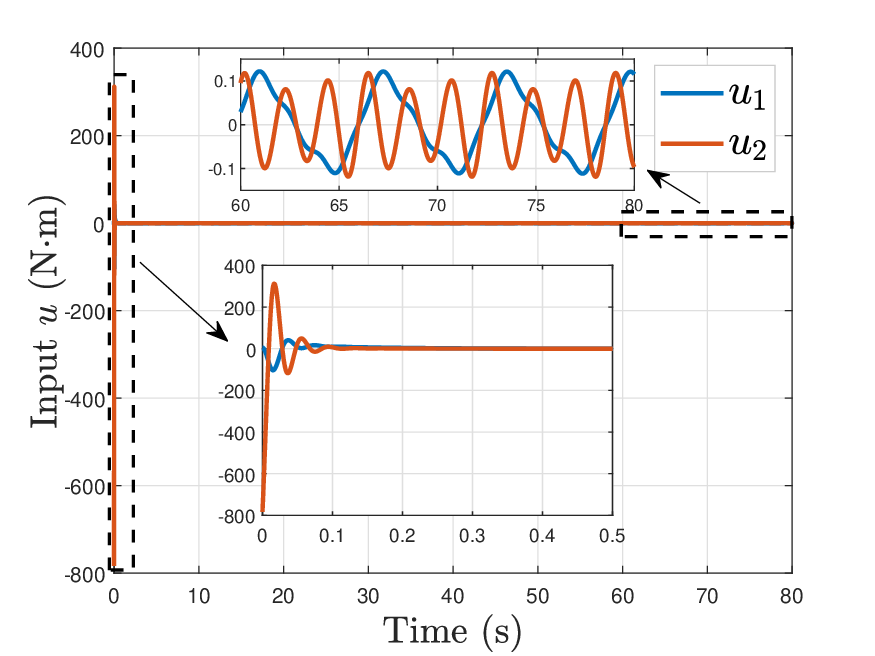}
         \caption{Control input $u = [u_{1}, u_{2}]^{\t}$}
         \label{fig:u}
     \end{subfigure}
     \hfill
     \begin{subfigure}[t]{0.30\textwidth}
         \centering
         \includegraphics[width=\textwidth, trim={14 3 33 20}, clip]{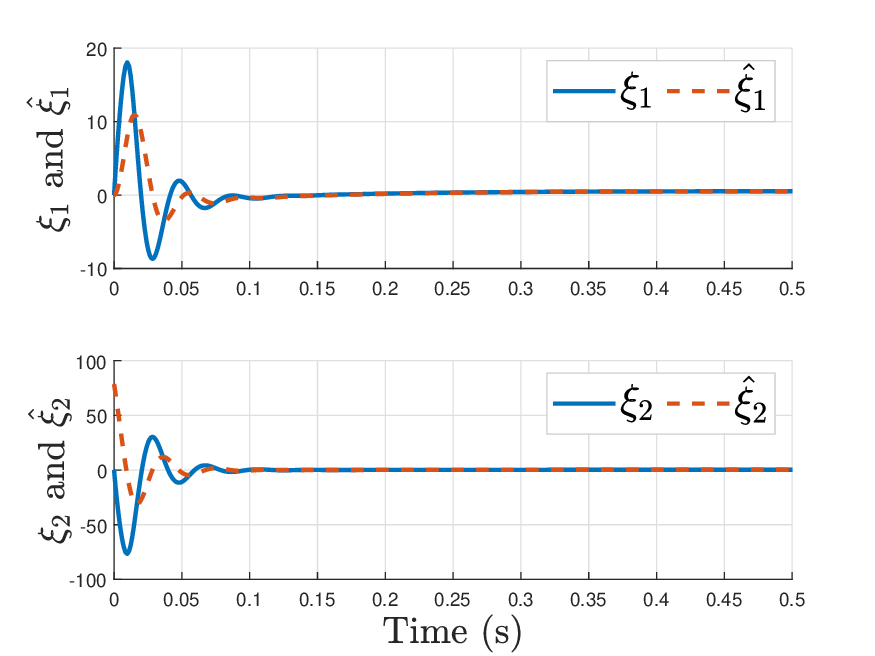}
         \caption{Joint velocity $\xi = [\xi_{1}, \xi_{2}]^{\t}$ and $\hat{\xi} = [\hat{\xi}_{1}, \hat{\xi}_{2}]^{\t}$}
         \label{fig:xi}
     \end{subfigure}
        \caption{Simulation results for the controller without saturation (controller \eqref{law-02}).}
        \label{fig:no-sat}
\end{figure*}

\begin{figure*}[!ht]
     \centering
     \begin{subfigure}[t]{0.30\textwidth}
         \centering
         \includegraphics[width=\textwidth, trim={14 3 33 20}, clip]{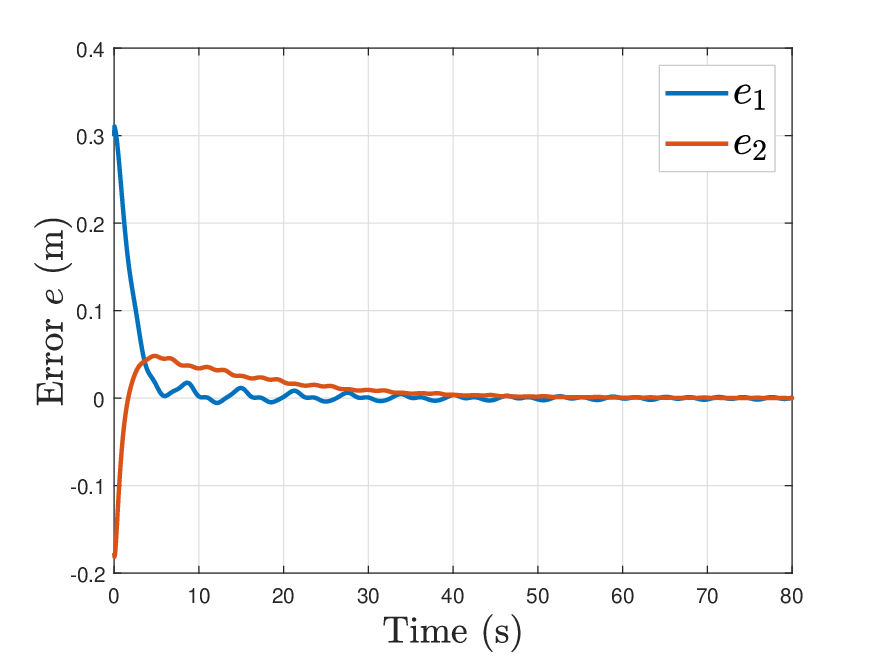}
         \caption{Error $e = [e_{1}, e_{2}]^{\t}$}
         \label{fig:sat-e}
     \end{subfigure}
     \hfill
     \begin{subfigure}[t]{0.30\textwidth}
         \centering
         \includegraphics[width=\textwidth, trim={14 3 33 20}, clip]{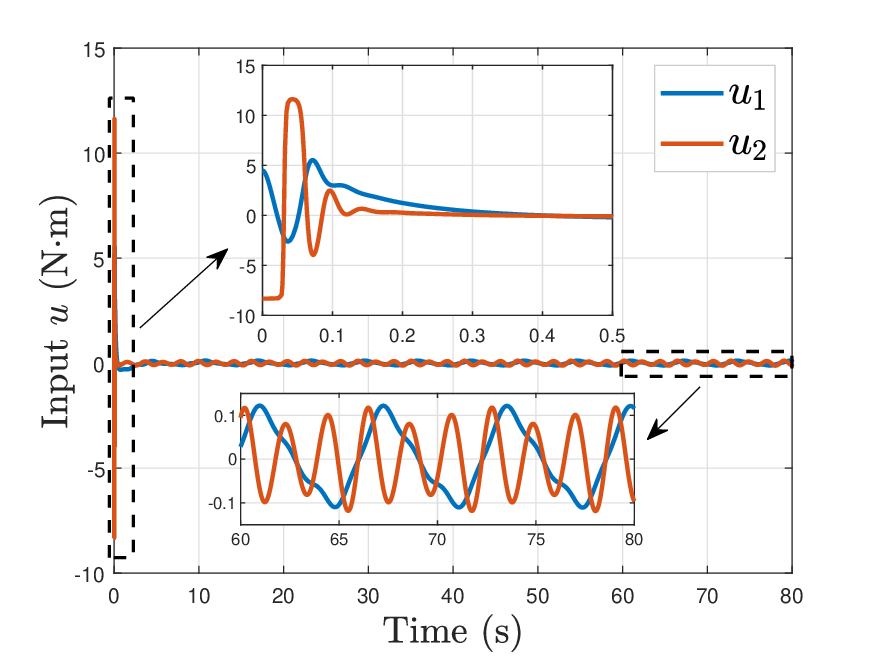}
         \caption{Control input $u = [u_{1}, u_{2}]^{\t}$}
         \label{fig:sat-u}
     \end{subfigure}
     \hfill
     \begin{subfigure}[t]{0.30\textwidth}
         \centering
         \includegraphics[width=\textwidth, trim={14 3 33 20}, clip]{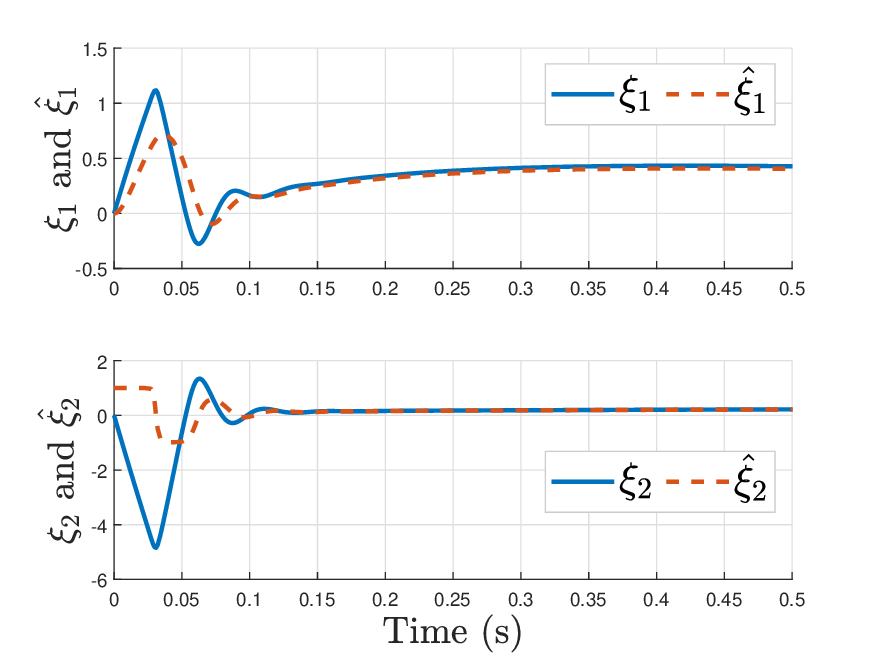}
         \caption{Joint velocity $\xi = [\xi_{1}, \xi_{2}]^{\t}$ and $\hat{\xi} = [\hat{\xi}_{1}, \hat{\xi}_{2}]^{\t}$}
         \label{fig:sat-xi}
     \end{subfigure}
        \caption{Simulation results for the controller with saturation (controller \eqref{law-02-sat}).}
        \label{fig:sat}
\end{figure*}

It should be noted that the proposed internal model-based velocity-free controller does
not rely on high-gain error feedback or high-gain observers, which are commonly
used in internal model-based output regulation designs, see for
example~\cite{Isidori2012robust}. Additionally, its design does not require a
prior knowledge of the boundaries of external disturbances. To close this
section, we present the following remarks concerning the proposed controller
\eqref{law-02}.  
First, in practice, we can construct the matrices $A_{1}$ and $A_{2}$ using the spectrum of $S$, which is based on \emph{a priori} knowledge of vibrations that we can gather from the environment where the robots are operated. For a practical way of constructing  $(A_{i},B_{i})$ for $i=1,2$, we refer to \cite{Bayu2008}. 
Second,the asymptotic convergence of $(e,\xi)$ is still ensured by modifying the stabilization part in \eqref{law-02} with saturation functions as follows
\begin{align}\label{law-02-sat}
        u &= -k_{p}J^{\t}(q)\frac{e}{1+e^{\t}e} - k_{d}\textnormal{Tanh}\big(\chi + hq\big) + g(q) \nonumber\\
        &\quad + B_{1}^{\t}(A_{1}\zeta_{1} - B_{1}q) +  J^{\t}(q)B_{2}^{\t}(A_{2}\zeta_{2} - B_{2}e)
\end{align}
where the vector function $\textnormal{Tanh}(\cdot)$ is defined as $\textnormal{Tanh}(x) = [\tanh(x_{1}),\dots,\tanh(x_{n})]^{\t}$ for all $x=[x_{1},\dots,x_{n}]^{\t}$. 
The proof can be completed by following the steps as in the proof of Proposition \ref{thm01} and using the following storage function: $U = \frac{1}{2}\bar{\zeta}^{\t}\bar{\zeta} + \frac{1}{2}k_{d}h^{-1}\sum_{i=1}^{n}\ln(\cosh(\hat{\xi}_{i})) + \frac{1}{2}k_{p}\ln(1 + e^{\t}e)  + \frac{1}{2}\xi^{\t}H(q)\xi$ where $\hat{\xi}_{i}$ is the $i$th element of $\hat{\xi} = \chi + hq$.

\section{Simulation result}\label{sec-sim}
To demonstrate the effectiveness of the proposed velocity-free controller, a
two-link planar manipulator is used for validation. We refer
to~\cite{Kelly1993simple} for the dynamic model and parameter setting of the
manipulator. The kinematic equation~\eqref{sys-kine} of the manipulator is
given by $x =f(q) = [l_{1}\cos(q_{1})+l_{2}\cos(q_{1}+q_{2}), l_{1}\sin(q_{1})+l_{2}\sin(q_{1}+q_{2})]^{\t}$, where $l_{i}$ for $i=1,2$ is the length of the
$i$th link and $q = [q_{1}, q_{2}]^{\t}$ is the joint angle vector. 
Correspondingly, the manipulator Jacobian matrix is $J(q) = \bbm{ -l_{1}\sin(q_{1})-l_{2}\sin(q_{1}+q_{2})  & -l_{2}\sin(q_{1}+q_{2}) \\
    l_{1}\cos(q_{1})+l_{2}\cos(q_{1}+q_{2})  &  l_{2}\cos(q_{1}+q_{2}) }$.

The external disturbances in \eqref{defn-d} are set as $d_{1} = 0.1 [\sin(\omega_{1}t), \sin(\omega_{3}t)]^{\t}$ and $d_{2} = 0.1[\sin(\omega_{2}t), \sin(\omega_{4}t)]^{\t}$ with known frequencies $\omega_{i}=i$ for $i=1,2,3,4$. The initial joint position and joint velocity are $q(0) = [0, \pi/4]^{\t}$ and $\xi(0) = [0, 0]^{\t}$, respectively. The end-effector is considered fixed at the end of the second link, and the desired end-effector position is chosen as $x_{d} = [0.064, 0.290]^{\t}$ in the robot base frame. 
Based on this setup, two simulations are conducted: 
\begin{itemize}[topsep=0pt,parsep=0pt,partopsep=0pt,itemsep=0pt]
\item [1)] In the first simulation, the velocity-free controller \eqref{law-02} proposed in Theorem~\ref{thm01} is used. Simulation results are shown in \Cref{fig:no-sat}. 
\item [2)] In the second simulation, we use the controller \eqref{law-02-sat} in which saturation functions are introduced to $e$ and $\hat{\xi}$. Simulation results are shown in \Cref{fig:sat}. 
\end{itemize}

The controller parameters for both simulations are selected as follows: $k_{p}=50$, $k_{d}=10$, $A_{i} = \textnormal{diag}\left(\sbm{0 & \omega_{i} \\ -\omega_{i} & 0}, \sbm{0 & \omega_{i+2} \\ -\omega_{i+2} & 0} \right)$, $B_{i} = \textnormal{diag}\left([1,0]^{\t}, [1,0]^{\t}\right)$, $i=1,2$, and $h =100$. 
The initial states of the internal models and the filter are all zero.

\Cref{fig:e} and \Cref{fig:sat-e} demonstrate that, in both cases, the regulation error $e$ converges to zero as expected despite the presence of external disturbances. A comparison between  \Cref{fig:u} and \Cref{fig:sat-u} shows that the control input of the first simulation can peak to large values during an initial transient period, whereas the control input of the second simulation is limited to a more acceptable level due to the use of saturation functions. It should be noted that the steady-state input signals of both simulations are sinusoidal waves capable of counteracting the effect of the external disturbances. \Cref{fig:xi} and \Cref{fig:sat-xi} demonstrate the joint velocity $\xi$ and the output of the filter $\hat{\xi} = \chi + hq$ of the two simulations. \Cref{fig:xi} illustrates the peaking phenomenon of the filter without saturation. \Cref{fig:sat-xi} shows that the peaking effect is reduced by saturating the estimates.

\section{Conclusions}\label{sec-con}
This paper considered task-space regulation of robot manipulators subject to sinusoidal external disturbances with known frequencies. We developed both full-state feedback and velocity-free controllers utilizing tools from internal model-based and passivity-based approaches. The proposed controllers ensure complete disturbance rejection and guarantee asymptotic convergence of the regulation error to zero. 

Relating to the current research, there comes up with an interesting question
of a future study on velocity-free regulation problems with uncertain
exosystems. We shall note that it would never be a trivial task because of the
technical challenge in constructing suitable control Lyapunov functions when
nonlinear or adaptive internal models (see~\cite{Bin2022internal} for a quick
overview) were incorporated. In this regard, one may refer
to~\cite{Loria2016observers} for a recent study of tracking control in joint
space by output feedback. To some extent, one must overcome the above hurdle
due to the indispensable role of internal models for the problem when merely
using output feedback. Furthermore, an extension of the proposed approach will
be explored to investigate the tracking of exogenous signals~\cite{Wu2025task}
for robotic systems without using velocity measurements.

\appendix
\section{Appendix: Proof of Lemma~\ref{lem-app}} 
We prove the result by using the PBH observability test \cite[p.~366]{Kailath1980-book}: The pair $(\bfA,\bfC)$ will be observable if and only if the matrix $\bbm{sI-\bfA \\ \bfC}$ has rank $n_{1}+n_{2}$ for all $s$. The proof falls naturally into three parts.

\begin{itemize}
    
\item [1)] For all $s\notin\sigma(\bfA_{1})\cap\sigma(\bfA_{2})$, we have
\begin{align*}
    \textnormal{rank} \bbm{sI - \bfA \\ \bfC} =
    \textnormal{rank} \bbm{sI - \bfA_{1} & 0 \\ 0 & sI - \bfA_{2} \\ \bfC_{1} & \bfT \bfC_{2}} = n_{1} + n_{2}
\end{align*}
where $\sigma(\bfA_{1})$ and $\sigma(\bfA_{2})$ denote the sets of eigenvalues of $\bfA_{1}$ and $\bfA_{2}$, respectively.

\item [2)] For all $s\in\sigma(\bfA_{1})$, taking into account $\sigma(\bfA_{1})\cap\sigma(\bfA_{2})=\emptyset$, we have $\textnormal{rank}(\bbm{sI - \bfA_{2}}) = n_{2}$. Hence,  
\begin{align*}
    \textnormal{rank} \bbm{sI - \bfA \\ \bfC} &=
    \textnormal{rank} \bbm{sI - \bfA_{1} & 0 \\ 0 & sI - \bfA_{2} \\ \bfC_{1} & \bfT \bfC_{2}} \\
    &= \textnormal{rank} \bbm{sI - \bfA_{1}  \\ \bfC_{1}} + n_{2} = n_{1} + n_{2}.
\end{align*}

\item [3)] Similarly, for all $s\in\sigma(\bfA_{2})$, 
\begin{align*}
    \textnormal{rank} \bbm{sI - \bfA \\ \bfC} 
    &= n_{1} + \textnormal{rank} \bbm{sI - \bfA_{2}  \\ \bfT \bfC_{2}} \\
    &= n_{1} + \textnormal{rank} (\bbm{I & 0 \\ 0 & \bfT}\bbm{sI - \bfA_{2} \\ \bfC_{2}}) .
\end{align*}
Since $\bfT$ has full column rank, we have 
\begin{align*}
    \textnormal{rank} (\bbm{I & 0 \\ 0 & \bfT}\bbm{sI - \bfA_{2} \\ \bfC_{2}}) 
    = \textnormal{rank} \bbm{sI - \bfA_{2} \\ \bfC_{2}}
\end{align*}
Then,
\begin{align*}
    \textnormal{rank} \bbm{sI - \bfA \\ \bfC} 
    &= n_{1} + \textnormal{rank} \bbm{sI - \bfA_{2} \\ \bfC_{2}} = n_{1} + n_{2}.
\end{align*}

\end{itemize}

Finally, we can conclude that the matrix $\bbm{sI - \bfA \\ \bfC}$ has rank $n_{1}+n_{2}$ for all $s$, which implies that $(\bfA,\bfC)$ is observable. The proof is complete.


\end{document}